\documentclass[11pt]{article}
\usepackage[a4paper]{geometry}
\usepackage{amsfonts, amsmath, amssymb, amsthm, graphicx, caption, authblk, multirow, makecell, framed, float, xcolor, enumitem, tikz, hyperref}
\theoremstyle{definition}

\newtheorem{lemma}{Lemma}

\theoremstyle{remark}

\definecolor{blk}{RGB}{63,63,63}
\newcommand*{\mybox}[1]{%
  \framebox{\raisebox{0cm}[0.5\baselineskip][0.05\baselineskip]{%
    \hbox to 0.10cm {\hss#1\hss}}}\hspace{0.05cm}}

\begin{document}
\title{Physical ZKP for Makaro Using a Standard Deck of Cards}
\author[1]{Suthee Ruangwises\thanks{\texttt{ruangwises@gmail.com}}}
\author[1]{Toshiya Itoh\thanks{\texttt{titoh@c.titech.ac.jp}}}
\affil[1]{Department of Mathematical and Computing Science, Tokyo Institute of Technology, Tokyo, Japan}
\date{}
\maketitle

\begin{abstract}
Makaro is a logic puzzle with an objective to fill numbers into a rectangular grid to satisfy certain conditions. In 2018, Bultel et al. developed a physical zero-knowledge proof (ZKP) protocol for Makaro using a deck of cards, which allows a prover to physically convince a verifier that he/she knows a solution of the puzzle without revealing it. However, their protocol requires several identical copies of some cards, making it impractical as a deck of playing cards found in everyday life typically consists of all different cards. In this paper, we propose a new ZKP protocol for Makaro that can be implemented using a standard deck (a deck consisting of all different cards). Our protocol also uses asymptotically less cards than the protocol of Bultel et al. Most importantly, we develop a general method to encode a number with a sequence of all different cards. This allows us to securely compute several numerical functions using a standard deck, such as verifying that two given numbers are different and verifying that a number is the largest one among the given numbers.

\end{abstract}

\section{Introduction}
\textit{Makaro} is a logic puzzle created by Nikoli, a company that developed many famous logic puzzles including Sudoku and Kakuro. A Makaro puzzle consists of a rectangular grid of white and black cells. White cells are divided into polyominoes called \textit{rooms}, with some cells already containing a number, while each black cell contain an arrow pointing to some direction. The objective of this puzzle is to fill a number into each empty white cell according to the following rules \cite{nikoli}.
\begin{enumerate}
	\item \textit{Room condition}: Each room must contain consecutive numbers starting from 1 to its \textit{size} (the number of cells in the room).
	\item \textit{Neighbor condition}: Two (horizontally or vertically) adjacent cells in different rooms must contain different numbers.
	\item \textit{Arrow condition}: Each arrow in a black cell must point to the only largest number among the (up to) four numbers in the white cells adjacent to that black cell. See Fig. \ref{fig1}.
\end{enumerate}

\begin{figure}
\centering
\begin{tikzpicture}
\filldraw[draw=blk,fill=blk] (2.4,0) rectangle (3.2,0.8);
\filldraw[draw=blk,fill=blk] (0.8,1.6) rectangle (1.6,2.4);
\filldraw[draw=blk,fill=blk] (3.2,1.6) rectangle (4,2.4);
\filldraw[draw=blk,fill=blk] (1.6,2.4) rectangle (2.4,3.2);
\filldraw[draw=blk,fill=blk] (1.6,3.2) rectangle (2.4,4);

\draw[step=0.8cm,color={rgb:black,1;white,4}] (0,0) grid (4,4);

\draw[line width=0.6mm] (0,0) -- (0,4);
\draw[line width=0.6mm] (0.8,0) -- (0.8,4);
\draw[line width=0.6mm] (1.6,0.8) -- (1.6,4);
\draw[line width=0.6mm] (2.4,0) -- (2.4,0.8);
\draw[line width=0.6mm] (2.4,1.6) -- (2.4,4);
\draw[line width=0.6mm] (3.2,0) -- (3.2,0.8);
\draw[line width=0.6mm] (3.2,1.6) -- (3.2,2.4);
\draw[line width=0.6mm] (4,0) -- (4,4);
\draw[line width=0.6mm] (0,0) -- (4,0);
\draw[line width=0.6mm] (1.6,0.8) -- (3.2,0.8);
\draw[line width=0.6mm] (0,1.6) -- (1.6,1.6);
\draw[line width=0.6mm] (2.4,1.6) -- (4,1.6);
\draw[line width=0.6mm] (0.8,2.4) -- (2.4,2.4);
\draw[line width=0.6mm] (3.2,2.4) -- (4,2.4);
\draw[line width=0.6mm] (1.6,3.2) -- (2.4,3.2);
\draw[line width=0.6mm] (0,4) -- (4,4);

\node at (3.6,0.4) {1};
\node at (0.4,2.8) {3};
\node at (3.6,3.6) {2};

\node at (2.8,0.4) {\textcolor{white}{\boldmath$\Uparrow$}};
\node at (1.2,2) {\textcolor{white}{\boldmath$\Downarrow$}};
\node at (3.6,2) {\textcolor{white}{\boldmath$\Downarrow$}};
\node at (2,2.8) {\textcolor{white}{\boldmath$\Rightarrow$}};
\node at (2,3.6) {\textcolor{white}{\boldmath$\Leftarrow$}};
\end{tikzpicture}
\hspace{1.5cm}
\begin{tikzpicture}
\filldraw[draw=blk,fill=blk] (2.4,0) rectangle (3.2,0.8);
\filldraw[draw=blk,fill=blk] (0.8,1.6) rectangle (1.6,2.4);
\filldraw[draw=blk,fill=blk] (3.2,1.6) rectangle (4,2.4);
\filldraw[draw=blk,fill=blk] (1.6,2.4) rectangle (2.4,3.2);
\filldraw[draw=blk,fill=blk] (1.6,3.2) rectangle (2.4,4);

\draw[step=0.8cm,color={rgb:black,1;white,4}] (0,0) grid (4,4);

\draw[line width=0.6mm] (0,0) -- (0,4);
\draw[line width=0.6mm] (0.8,0) -- (0.8,4);
\draw[line width=0.6mm] (1.6,0.8) -- (1.6,4);
\draw[line width=0.6mm] (2.4,0) -- (2.4,0.8);
\draw[line width=0.6mm] (2.4,1.6) -- (2.4,4);
\draw[line width=0.6mm] (3.2,0) -- (3.2,0.8);
\draw[line width=0.6mm] (3.2,1.6) -- (3.2,2.4);
\draw[line width=0.6mm] (4,0) -- (4,4);
\draw[line width=0.6mm] (0,0) -- (4,0);
\draw[line width=0.6mm] (1.6,0.8) -- (3.2,0.8);
\draw[line width=0.6mm] (0,1.6) -- (1.6,1.6);
\draw[line width=0.6mm] (2.4,1.6) -- (4,1.6);
\draw[line width=0.6mm] (0.8,2.4) -- (2.4,2.4);
\draw[line width=0.6mm] (3.2,2.4) -- (4,2.4);
\draw[line width=0.6mm] (1.6,3.2) -- (2.4,3.2);
\draw[line width=0.6mm] (0,4) -- (4,4);

\node at (0.4,0.4) {2};
\node at (1.2,0.4) {1};
\node at (2,0.4) {2};
\node at (3.6,0.4) {1};
\node at (0.4,1.2) {1};
\node at (1.2,1.2) {3};
\node at (2,1.2) {4};
\node at (2.8,1.2) {3};
\node at (3.6,1.2) {5};
\node at (0.4,2) {2};
\node at (2,2) {2};
\node at (2.8,2) {4};
\node at (0.4,2.8) {3};
\node at (1.2,2.8) {1};
\node at (2.8,2.8) {5};
\node at (3.6,2.8) {3};
\node at (0.4,3.6) {1};
\node at (1.2,3.6) {2};
\node at (2.8,3.6) {1};
\node at (3.6,3.6) {2};

\node at (2.8,0.4) {\textcolor{white}{\boldmath$\Uparrow$}};
\node at (1.2,2) {\textcolor{white}{\boldmath$\Downarrow$}};
\node at (3.6,2) {\textcolor{white}{\boldmath$\Downarrow$}};
\node at (2,2.8) {\textcolor{white}{\boldmath$\Rightarrow$}};
\node at (2,3.6) {\textcolor{white}{\boldmath$\Leftarrow$}};
\end{tikzpicture}
\caption{An example of a Makaro puzzle (left) and its solution (right)}
\label{fig1}
\end{figure}

Determining whether a given Makaro puzzle has a solution has been proved to be NP-complete \cite{np}.

Suppose that Amber created a difficult Makaro puzzle and challenged her friend Bennett to solve it. After a while, Bennett could not solve her puzzle and began to doubt whether the puzzle has a solution. Amber needs to convince him that her puzzle actually has a solution without revealing it to him. In this situation, Amber needs a \textit{zero-knowledge proof (ZKP)}.

\subsection{Zero-Knowledge Proof}
The concept of a ZKP was first introduced by Goldwasser et al. \cite{zkp0}. A ZKP is an interactive proof between $P$ and $V$ where both of them are given a computational problem $x$, but only $P$ knows a solution $w$ of $x$. A ZKP with perfect completeness and perfect soundness must satisfy the following three properties.

\begin{enumerate}
	\item \textbf{Perfect Completeness:} If $P$ knows $w$, then $V$ always accepts.
	\item \textbf{Perfect Soundness:} If $P$ does not know $w$, then $V$ always rejects.
	\item \textbf{Zero-knowledge:} $V$ learns nothing about $w$. Formally, there exists a probabilistic polynomial time algorithm $S$ (called a \textit{simulator}) that does not know $w$ but has access to $V$, and the outputs of $S$ follow the same probability distribution as the ones from the actual protocol.
\end{enumerate}

Many recent results have been focusing on constructing physical ZKPs using objects found in everyday life such as a deck of cards and envelopes. These physical protocols have benefits that they do not require computers and also allow external observers to verify that the prover truthfully executes the protocol (which is often a challenging task for digital protocols). They are also suitable for teaching purpose and can be used to teach the concept of a ZKP to non-experts.

\subsection{Related Work}
\subsubsection{Protocol of Bultel et al.}
In 2018, Bultel et al. \cite{makaro} developed the first card-based ZKP protocol for Makaro. Their protocol uses $\Theta(nk)$ cards, where $n$ and $k$ are the number of white cells and the size of the largest room, respectively. However, it requires $\Theta(nk)$ identical copies of a specific card (and also $\Theta(n)$ identical copies of another card).

As a deck of playing cards found in everyday life typically consists of all different cards, $\Theta(nk)$ identical decks are actually required to implement this protocol, making the protocol very impractical. Another option is to use a different kind of deck (e.g. cards from board games) that contains several identical copies of some cards, but these decks are more difficult to find in everyday life.

\subsubsection{Other Protocols}
Besides Makaro, card-based ZKP protocols for many other logic puzzles have also been developed: Sudoku \cite{sudoku0,sudoku}, Akari \cite{akari}, Takuzu \cite{akari,takuzu}, Kakuro \cite{akari,kakuro}, KenKen \cite{akari}, Norinori \cite{norinori}, Slitherlink \cite{slitherlink}, Juosan \cite{takuzu}, Numberlink \cite{numberlink}, Suguru \cite{suguru}, Ripple Effect \cite{ripple}, Nurikabe \cite{nurikabe}, Hitori \cite{nurikabe}, Bridges \cite{bridges}, Masyu \cite{slitherlink}, Nonogram \cite{nonogram2}, Heyawake \cite{nurikabe}, and Shikaku \cite{shikaku}. All of these protocols, however, require a deck with repeated cards.

An open problem to develop ZKP protocols for logic puzzles using a standard deck (a deck consisting of all different cards) was posed by Koyama et al. \cite{standard4}. This problem was recently answered by Ruangwises \cite{sudoku2}, who developed a ZKP protocol for Sudoku using a standard deck, the first standard deck protocol for any kind of logic puzzle. However, the protocol in \cite{sudoku2} was specifically designed to tackle only the rules of Sudoku and cannot be applied to verify other numerical functions or other logic puzzles, thus having limited utility.

Other than logic puzzles, card-based protocols have also been widely studied in secure multi-party computation, a setting where multiple parties want to jointly compute a function of their secret inputs without revealing them. Almost all of existing protocols, however, also use a deck with repeated cards. The only exceptions are \cite{standard3,standard4,standardyao,standard2,standard1} which proposed AND, XOR, copy, and Yao's millionaire protocols using a standard deck.

\subsection{Our Contribution}
Considering the drawback of the protocol of Bultel et al. \cite{makaro}, we aim to develop a more practical ZKP protocol for Makaro that can be implemented using a standard deck.\footnote{Although a ``standard deck'' of playing cards found in everyday life typically consists of 52 different cards, in theory we study a general setting where the deck is arbitrarily large, consisting of all different cards.}

In this paper, we propose a new ZKP protocol for Makaro with perfect completeness and soundness using a standard deck. It is also the second standard deck protocol for any logic puzzle, after the one for Sudoku \cite{sudoku2}. Remarkably, our protocol uses asymptotically less cards than the protocol of Bultel et al. (see Table \ref{table1}). This is a noteworthy achievement as card-based protocols that use a standard deck generally require more cards than their counterparts that use a deck with repeated cards \cite{anydeck}. (In particular, the standard deck protocol for Sudoku \cite{sudoku2} also requires more cards than its counterpart \cite{sudoku}.)

\begin{table}
	\centering
	\begin{tabular}{|c|c|c|c|}
		\hline
		\textbf{Protocol} & \textbf{\thead{Standard\\ Deck?}} & \textbf{\#Cards} \\ \hline
		\textbf{Bultel et al. \cite{makaro}} & no & $\Theta(nk)$ \\ \hline
		\textbf{Ours} & yes & $\Theta(n+k)$ \\ \hline
	\end{tabular}
	\medskip
	\caption{The number of required cards for each protocol for Makaro, where $n$ and $k$ are the number of white cells and the size of the largest room, respectively} \label{table1}
\end{table}

Most importantly, we develop a general method to encode a number with a sequence of all different cards. This allows us to securely compute several numerical functions using a standard deck, such as verifying that two given numbers are different and verifying that a number is the largest one among the given numbers.

\section{Preliminaries}
Let $n$ be the number of white cells and $k$ be the size of the largest room in the Makaro grid.

We assume that all cards used in our protocols have different front sides and identical back sides. For didactic purpose, cards are divided into \textit{sets}. Cards in the same set are denoted by the same letter with different indices, e.g. cards $a_1,a_2,a_3,a_4$ are in the same set.

In an $\ell \times m$ \textit{matrix} of cards, let Row $i$ denote the $i$-th topmost row, and Column $j$ denote the $j$-th leftmost column.

\subsection{Pile-Shifting Shuffle}
Given an $\ell \times m$ matrix of cards, a \textit{pile-shifting shuffle} \cite{polygon} rearranges the columns of the matrix by a random cyclic shift unknown to all parties. It can be implemented in real world by putting the cards in each column into an envelope and then taking turns to apply \textit{Hindu cuts} (taking several envelopes from the bottom and putting them on the top) to the pile of envelopes \cite{hindu}.

\subsection{Pile-Scramble Shuffle}
Given an $\ell \times m$ matrix of cards, a \textit{pile-scramble shuffle} \cite{scramble} rearranges the columns of the matrix by a random permutation unknown to all parties. It can be implemented in real world by putting the cards in each column into an envelope and then jointly scrambling the envelopes together randomly.

\section{Main Protocol}
\subsection{Cell Cards}
We use a \textit{cell card} to represent each white cell in the grid. Cells in the same room are represented by cards in the same set. To avoid confusion, a cell card is always denoted by a Greek letter followed by an index equal to the number in the cell it represents. We have cell cards in sets $\alpha_i,\beta_i,\gamma_i,...$ and so on. See Fig. \ref{fig4} for an example.\footnote{Assume that we have $\ell$ cards with different numbers, e.g. cards with numbers $1,2,...,\ell$. In the example in Fig. \ref{fig4}, we can, for instance, regard cards $1,2,3$ on cells with numbers $1,2,3$ in the top-left room as $\alpha_1,\alpha_2,\alpha_3$, cards $4,5$ on cells with numbers $1,2$ in the top-center room as $\beta_1,\beta_2$, cards $6,7,8,9,10$ on cells with numbers $1,2,3,4,5$ in the top-right room as $\gamma_1,\gamma_2,\gamma_3,\gamma_4,\gamma_5$, and so on.}

\begin{figure}
\centering
\begin{tikzpicture}
\filldraw[draw=blk,fill=blk] (2.4,0) rectangle (3.2,0.8);
\filldraw[draw=blk,fill=blk] (0.8,1.6) rectangle (1.6,2.4);
\filldraw[draw=blk,fill=blk] (3.2,1.6) rectangle (4,2.4);
\filldraw[draw=blk,fill=blk] (1.6,2.4) rectangle (2.4,3.2);
\filldraw[draw=blk,fill=blk] (1.6,3.2) rectangle (2.4,4);

\draw[step=0.8cm,color={rgb:black,1;white,4}] (0,0) grid (4,4);

\draw[line width=0.6mm] (0,0) -- (0,4);
\draw[line width=0.6mm] (0.8,0) -- (0.8,4);
\draw[line width=0.6mm] (1.6,0.8) -- (1.6,4);
\draw[line width=0.6mm] (2.4,0) -- (2.4,0.8);
\draw[line width=0.6mm] (2.4,1.6) -- (2.4,4);
\draw[line width=0.6mm] (3.2,0) -- (3.2,0.8);
\draw[line width=0.6mm] (3.2,1.6) -- (3.2,2.4);
\draw[line width=0.6mm] (4,0) -- (4,4);
\draw[line width=0.6mm] (0,0) -- (4,0);
\draw[line width=0.6mm] (1.6,0.8) -- (3.2,0.8);
\draw[line width=0.6mm] (0,1.6) -- (1.6,1.6);
\draw[line width=0.6mm] (2.4,1.6) -- (4,1.6);
\draw[line width=0.6mm] (0.8,2.4) -- (2.4,2.4);
\draw[line width=0.6mm] (3.2,2.4) -- (4,2.4);
\draw[line width=0.6mm] (1.6,3.2) -- (2.4,3.2);
\draw[line width=0.6mm] (0,4) -- (4,4);

\node at (0.4,0.4) {$\epsilon_2$};
\node at (1.2,0.4) {$\zeta_1$};
\node at (2,0.4) {$\zeta_2$};
\node at (3.6,0.4) {$\delta_1$};
\node at (0.4,1.2) {$\epsilon_1$};
\node at (1.2,1.2) {$\zeta_3$};
\node at (2,1.2) {$\delta_4$};
\node at (2.8,1.2) {$\delta_3$};
\node at (3.6,1.2) {$\delta_5$};
\node at (0.4,2) {$\alpha_2$};
\node at (2,2) {$\delta_2$};
\node at (2.8,2) {$\gamma_4$};
\node at (0.4,2.8) {$\alpha_3$};
\node at (1.2,2.8) {$\beta_1$};
\node at (2.8,2.8) {$\gamma_5$};
\node at (3.6,2.8) {$\gamma_3$};
\node at (0.4,3.6) {$\alpha_1$};
\node at (1.2,3.6) {$\beta_2$};
\node at (2.8,3.6) {$\gamma_1$};
\node at (3.6,3.6) {$\gamma_2$};

\node at (2.8,0.4) {\textcolor{white}{\boldmath$\Uparrow$}};
\node at (1.2,2) {\textcolor{white}{\boldmath$\Downarrow$}};
\node at (3.6,2) {\textcolor{white}{\boldmath$\Downarrow$}};
\node at (2,2.8) {\textcolor{white}{\boldmath$\Rightarrow$}};
\node at (2,3.6) {\textcolor{white}{\boldmath$\Leftarrow$}};
\end{tikzpicture}
\caption{A cell card representing each white cell in the solution of the puzzle in Fig. \ref{fig1}}
\label{fig4}
\end{figure}

At the beginning, $P$ publicly places a face-down corresponding cell card on each white cell already having a number. Then, $P$ secretly places a face-down corresponding cell card according to his/her solution on each empty white cell.

\subsection{Verifying Room Condition} \label{room}
Consider a room $R$ of size $p$ in the Makaro grid containing cells represented by cell cards $\alpha_1,\alpha_2,...,\alpha_p$. This subprotocol allows $P$ to show that the cell cards in $R$ consist of a permutation of $\alpha_1,\alpha_2,...,\alpha_p$ without revealing their order. It was developed by Sasaki et al. \cite{sudoku}.

Besides cell cards, we also use \textit{helping cards} $h_i$ ($i=1,2,...,k$) in our protocol.

\begin{figure}
\centering
\begin{tikzpicture}
\node at (0.0,1.4) {\mybox{?}};
\node at (0.6,1.4) {\mybox{?}};
\node at (1.2,1.4) {...};
\node at (1.8,1.4) {\mybox{?}};

\node at (0.05,1) {$\alpha_?$};
\node at (0.65,1) {$\alpha_?$};
\node at (1.85,1) {$\alpha_?$};

\node at (0.0,0.4) {\mybox{?}};
\node at (0.6,0.4) {\mybox{?}};
\node at (1.2,0.4) {...};
\node at (1.8,0.4) {\mybox{?}};

\node at (0.05,0) {$h_1$};
\node at (0.65,0) {$h_2$};
\node at (1.85,0) {$h_p$};
\end{tikzpicture}
\caption{A $2 \times p$ matrix constructed in Step 2}
\label{fig5}
\end{figure}

\begin{enumerate}
	\item Take all cell cards in $R$ in any specific order (e.g. from top to bottom, then from left to right) and place them face-down in Row 1 of a matrix $M$.
	\item Publicly place face-down helping cards $h_1,h_2,...,h_p$ in Row 2 of $M$ in this order from left to right. $M$ is now a $2 \times p$ matrix (see Fig. \ref{fig5}).
	\item Apply the pile-scramble shuffle to $M$.
	\item Turn over all cards in Row 1 of $M$. If the sequence is a permutation of $\alpha_1,\alpha_2,...,\alpha_p$, proceed to the next step; otherwise, $V$ rejects.
	\item Turn over all face-up cards in $M$. Apply the pile-scramble shuffle to $M$ again.
	\item Turn over all cards in Row 2 of $M$. Arrange the columns of $M$ such that the cards in Row 2 are $h_1,h_2,...,h_p$ in this order from left to right. Note that the columns of $M$ are now reverted to their original order.
	\item Take the cards in Row 1 of $M$ and place them back into room $R$ in the same order we take them in Step 1.
\end{enumerate}

$P$ applies this subprotocol for every room in the Makaro grid to verify the room condition.

However, verifying the neighbor condition and arrow condition is more difficult and cannot be done by using cell cards alone. Therefore, we have to develop a method to encode a number with a sequence of all different cards.

\subsection{Encoding Sequences}
In previous ZKP protocols for other logic puzzles \cite{makaro,suguru,numberlink,ripple,bridges}, a number $x$ ($1 \leq x \leq m$) is often encoded by a sequence $E_m(x)$ of $m$ consecutive cards, with all of them being \mybox{$\clubsuit$}s except the $x$-th leftmost card being a \mybox{$\heartsuit$} (e.g. $E_4(2)$ is \mybox{$\clubsuit$}\mybox{$\heartsuit$}\mybox{$\clubsuit$}\mybox{$\clubsuit$}). We will employ that idea to develop an encoding sequence for a number $x$ using all different cards.

Besides cell cards and helping cards, we also use \textit{encoding cards} $a_i,b_i,c_i,d_i$ ($i=1,2,...,2k-1$) in our protocol. (We need four sets of encoding cards because we later have to compare up to four numbers at the same time during the arrow condition verification.)

For a fixed integer $m \leq 2k-1$, define a sequence $E_m^a(x)$ to be a sequence of $m$ consecutive cards, where the $x$-th leftmost card is $a_1$, and the other $m-1$ cards are a uniformly random permutation of $a_2,a_3,...,a_m$ unknown to $V$.

The role of the card $a_1$ in $E_m^a(x)$ is to mark the value of $x$, similarly to a \mybox{$\heartsuit$} in $E_m(x)$. Note that the order of $a_2,a_3,...,a_m$ must be unknown to $V$ in order for the protocol to be zero-knowledge, so each encoding sequence is for one-time use only.

We also define sequences $E_m^b(x)$, $E_m^c(x)$, and $E_m^d(x)$ analogously, using encoding cards from sets $b_i$, $c_i$, and $d_i$, with cards $b_1$, $c_1$, and $d_1$ as marking points, respectively.

\subsection{Conversion from Cell Cards to Encoding Sequences} \label{conversion}
This is the most crucial subprotocol in our protocol. Let $w$ be any white cell represented by a cell card $\alpha_x$. Suppose that $w$ is located in a room $R$ with size $p$. This subprotocol allows $P$ to construct an encoding sequence $E_m^a(x)$ for some fixed $m \geq p$ without revealing the value $x$ to $V$, while leaving all cell cards in $R$ unchanged.

\begin{figure}
\centering
\begin{tikzpicture}
\node at (0.0,2.4) {\mybox{?}};
\node at (0.6,2.4) {\mybox{?}};
\node at (1.2,2.4) {...};
\node at (1.8,2.4) {\mybox{?}};
\node at (2.4,2.4) {\mybox{?}};
\node at (3.0,2.4) {\mybox{?}};
\node at (3.6,2.4) {...};
\node at (4.2,2.4) {\mybox{?}};

\node at (0.0,2) {$\alpha_?$};
\node at (0.6,2) {$\alpha_?$};
\node at (1.8,2) {$\alpha_?$};
\node at (2.4,2) {$\alpha_x$};
\node at (3.0,2) {$\alpha_?$};
\node at (4.2,2) {$\alpha_?$};

\node at (0.0,1.4) {\mybox{?}};
\node at (0.6,1.4) {\mybox{?}};
\node at (1.2,1.4) {...};
\node at (1.8,1.4) {\mybox{?}};
\node at (2.4,1.4) {\mybox{?}};
\node at (3.0,1.4) {\mybox{?}};
\node at (3.6,1.4) {...};
\node at (4.2,1.4) {\mybox{?}};

\node at (0.0,1) {$h_1$};
\node at (0.6,1) {$h_2$};
\node at (1.8,1) {$h_{i-1}$};
\node at (2.4,1) {$h_i$};
\node at (3.0,1) {$h_{i+1}$};
\node at (4.2,1) {$h_p$};

\node at (0.0,0.4) {\mybox{?}};
\node at (0.6,0.4) {\mybox{?}};
\node at (1.2,0.4) {...};
\node at (1.8,0.4) {\mybox{?}};
\node at (2.4,0.4) {\mybox{?}};
\node at (3.0,0.4) {\mybox{?}};
\node at (3.6,0.4) {...};
\node at (4.2,0.4) {\mybox{?}};

\node at (0.0,0) {$a_?$};
\node at (0.6,0) {$a_?$};
\node at (1.8,0) {$a_?$};
\node at (2.4,0) {$a_1$};
\node at (3.0,0) {$a_?$};
\node at (4.2,0) {$a_?$};
\end{tikzpicture}
\caption{A $3 \times k$ matrix $M$ constructed in Step 5}
\label{fig6}
\end{figure}

\begin{enumerate}
	\item Take all cell cards in $R$ in any specific order (e.g. from top to bottom, then from left to right) and place them face-down in Row 1 of a matrix $M$. Suppose the card $\alpha_x$ is located at Column $i$ of $M$.
	\item Publicly place face-down helping cards $h_1,h_2,...,h_p$ in Row 2 of $M$ in this order from left to right.
	\item Publicly place face-down encoding card $a_1$ in Row 3 of $M$ at Column $i$.
	\item Secretly arrange face-down encoding cards $a_2,a_3,...,a_m$ in a uniformly random permutation unknown to $V$. Refer to this sequence as $S$.
	\item Take the $p-1$ leftmost cards of $S$ and place them in empty cells in Row 3 of $M$ in this order from left to right. Leave the $m-p$ rightmost cards of $S$ unchanged. $M$ is now a complete $3 \times p$ matrix (see Fig. \ref{fig6}).
	\item Apply the pile-scramble shuffle to $M$.
	\item Turn over all cards in Row 1 of $M$. Arrange the columns of $M$ such that the cards in Row 1 are $\alpha_1,\alpha_2,...,\alpha_p$ in this order from left to right.
	\item Take all cards in Row 3 of $M$ out of the matrix ($M$ now becomes a $2 \times p$ matrix). Refer to the sequence taken from Row 3 of $M$ as $T$. Append the $m-p$ rightmost cards of $S$ left in Step 5 to the right of $T$. The appended sequence is $E_m^a(x)$ as desired.
	\item Turn over all face-up cards in $M$. Apply the pile-scramble shuffle to $M$ again.
	\item Turn over all cards in Row 2 of $M$. Arrange the columns of $M$ such that the cards in Row 2 are $h_1,h_2,...,h_p$ in this order from left to right. Note that the columns of $M$ are now reverted to their original order.
	\item Take the cards in Row 1 of $M$ and place them back into room $R$ in the same order we take them in Step 1.
\end{enumerate}

\subsection{Verifying Neighbor Condition} \label{neighbor}
This subprotocol allows $P$ to show that two adjacent cells represented by $\alpha_x$ and $\beta_y$ in different rooms contain different numbers. The idea of this subprotocol is exactly the same as the one developed by Bultel et al. \cite[\S3.3 Step 2]{makaro} to verify the same condition, except that it uses encoding sequences $E_m^a(x)$ and $E_m^b(y)$ instead of $E_m(x)$ and $E_m(y)$.

Let $p$ and $q$ be the sizes of rooms containing $\alpha_x$ and $\beta_y$, respectively, and let $m=\max(p,q)$. Note that we have $m \leq k$. First, $P$ applies the conversion protocol in section \ref{conversion} to construct sequences $E_m^a(x)$ and $E_m^b(y)$ from $\alpha_x$ and $\beta_y$, respectively. Then, $P$ performs the following steps.

\begin{enumerate}
	\item Construct a $2 \times m$ matrix $M$ by placing $E_m^a(x)$ and $E_m^b(y)$ in Row 1 and Row 2, respectively.
	\item Apply the pile-scramble shuffle to $M$.
	\item Turn over all cards in Row 1 of $M$. Suppose $a_1$ is located at Column $i$.
	\item Turn over a card in Row 2 of $M$ at Column $i$. If it is not $b_1$, proceed to the next step; otherwise, $V$ rejects.
\end{enumerate}

$P$ applies this subprotocol for every pair of adjacent cells that are in different rooms in the Makaro grid to verify the neighbor condition.

\subsection{Verifying Arrow Condition} \label{arrow}
Suppose the (up to) four cells adjacent to a black cell containing an arrow are represented by $\alpha_x$, $\beta_y$, $\gamma_z$, and $\delta_t$, with an arrow pointing to $\alpha_x$.\footnote{Some of the cells may be in the same room, but this does not affect the conversion as we apply the conversion protocol to each cell card one by one.} This subprotocol allows $P$ to show that a number in the cell represented by $\alpha_x$ is the largest one among all numbers in these cells. The idea of this subprotocol is exactly the same as the one developed by Bultel et al. \cite[\S3.3 Step 3]{makaro} to verify the same condition, except that it uses encoding sequences $E_{2m-1}^a(x)$, $E_{2m-1}^b(y)$, $E_{2m-1}^c(z)$, and $E_{2m-1}^d(t)$ instead of $E_{2m-1}(x)$, $E_{2m-1}(y)$, $E_{2m-1}(z)$, and $E_{2m-1}(t)$.

Let $p$, $q$, $r$, and $s$ be the sizes of rooms containing $\alpha_x$, $\beta_y$, $\gamma_z$, and $\delta_t$, respectively, and let $m=\max(p,q,r,s)$. Note that we have $m \leq k$, and thus $2m-1 \leq 2k-1$. First, $P$ applies the conversion protocol in section \ref{conversion} to construct sequences $E_{2m-1}^a(x)$, $E_{2m-1}^b(y)$, $E_{2m-1}^c(z)$, and $E_{2m-1}^d(t)$ from $\alpha_x$, $\beta_y$, $\gamma_z$, and $\delta_t$, respectively. Then, $P$ performs the following steps.

\begin{enumerate}
	\item Construct a $4 \times (2m-1)$ matrix $M$ by placing $E_{2m-1}^a(x)$, $E_{2m-1}^b(y)$, $E_{2m-1}^c(z)$, and $E_{2m-1}^d(t)$ in Rows 1, 2, 3, and 4, respectively.
	\item Apply the pile-shifting shuffle to $M$.
	\item Turn over all cards in Row 1 of $M$. Suppose $a_1$ is located at Column $i$.
	\item Turn over cards in Rows 2, 3, and 4 of $M$ at Columns $i$, $i+1$, ..., $i+m-1$ (where the indices are taken modulo $2m-1$). If none of them is $b_1$, $c_1$, or $d_1$, proceed to the next step; otherwise, $V$ rejects.
\end{enumerate}

$P$ applies this subprotocol for every arrow in the Makaro grid to verify the arrow condition.

If the verification passes for all three conditions, then $V$ accepts.

\subsection{Complexity}
Our protocol uses $n$ cell cards, $k$ helping cards, and $4(2k-1)$ encoding cards, resulting in the total of $n+9k-4 = \Theta(n+k)$ cards. In comparison, the protocol of Bultel et al. \cite{makaro} requires $\Theta(nk)$ cards.

\section{Proof of Correctness and Security}
We will prove the perfect completeness, perfect soundness, and zero-knowledge properties of our protocol.

\begin{lemma}[Perfect Completeness] \label{lem1}
If $P$ knows a solution of the Makaro puzzle, then $V$ always accepts.
\end{lemma}

\begin{proof}
Suppose $P$ knows a solution and places cards on the grid accordingly.

First, we will prove the correctness of the conversion protocol in Section \ref{conversion}. From the way we construct the matrix $M$, in Step 5 the card $a_1$ is in the same column as $\alpha_x$, and the other $p-1$ cards in Row 3 are uniformly distributed among all $\frac{(m-1)!}{(m-p)!}$ permutations of $p-1$ cards selected from $a_2,a_3,...,a_m$. In Step 7, the card $a_1$ is moved to Column $x$. Hence, the appended sequence in Step 8 has $a_1$ as the $x$-th leftmost card, and the other $m-1$ cards are uniformly distributed among all $(m-1)!$ permutations of $a_2,a_3,...,a_m$ (which remains unknown to $V$). Therefore, the appended sequence is indeed $E_m^a(x)$.

Next, we will prove that the verification of all three conditions will pass.

\begin{itemize}
	\item For the room condition verification in Section \ref{room}, the cards that are turned over in Step 4 must be a permutation of $\alpha_1,\alpha_2,...,.\alpha_p$, so the verification will pass.
	\item For the neighbor condition verification in Section \ref{neighbor}, the cell cards are correctly converted to sequences $E_m^a(x)$ and $E_m^b(y)$. Since $x \neq y$, the cards $a_1$ and $b_1$ must be in different columns of $M$. Hence, the card that is turned over in Step 4 cannot be $b_1$, so the verification will pass.
	\item For the arrow condition verification in Section \ref{arrow}, the cell cards are correctly converted to sequences $E_{2m-1}^a(x)$, $E_{2m-1}^b(y)$, $E_{2m-1}^c(z)$, and $E_{2m-1}^d(t)$. Since $x$ is the only largest number among the four numbers, in Step 3 each of the cards $b_1$, $c_1$, and $d_1$ must be in one of Columns $i-1,i-2,...,i-m+1$ of $M$ (where the indices are taken modulo $2m-1$). Hence, the cards that are turned over in Step 4 cannot include $b_1$, $c_1$, or $d_1$, so the verification will pass.
\end{itemize}

Therefore, $V$ always accepts.
\end{proof}

\begin{lemma}[Perfect Soundness] \label{lem2}
If $P$ does not know a solution of the Makaro puzzle, then $V$ always rejects.
\end{lemma}

\begin{proof}
Suppose $P$ does not know a solution. At least one of the three conditions must be violated.
\begin{itemize}
	\item If the room condition is violated, consider the room condition verification in Section \ref{room} for a room that violates the condition. The cards that are turned over in Step 4 cannot be a permutation of $\alpha_1,\alpha_2,...,.\alpha_p$, so the verification will fail.
	\item If the neighbor condition is violated, consider the neighbor condition verification in Section \ref{neighbor} for a pair of adjacent cells that that violates the condition. We have $x=y$, so the cards $a_1$ and $b_1$ must be in the same column of $M$. Hence, the card that is turned over in Step 4 must be $b_1$, so the verification will fail.
	\item If the arrow condition is violated, consider the arrow condition verification in Section \ref{arrow} for an arrow that violates the condition. Suppose $y \geq x$. In Step 3, the card $b_1$ must be in one of Columns $i,i+1,...,i+m-1$ of $M$ (where the indices are taken modulo $2m-1$). Hence, the cards that are turned over in Step 4 must include $b_1$, so the verification will fail.
\end{itemize}

Therefore, $V$ always rejects.
\end{proof}

\begin{lemma}[Zero-Knowledge] \label{lem3}
During the verification, $V$ learns nothing about $P$'s solution.
\end{lemma}

\begin{proof}
It is sufficient to show that all distributions of cards that are turned face-up can be simulated by a simulator $S$ that does not know $P$'s solution.

\begin{itemize}
	\item In the room condition verification in Section \ref{room}:
	\begin{itemize}
		\item In Step 4, the orders of face-up cards are uniformly distributed among all $p!$ permutations of $\alpha_1,\alpha_2,...,\alpha_p$, so it can be simulated by $S$.
		\item In Step 6, the orders of face-up cards are uniformly distributed among all $p!$ permutations of $h_1,h_2,...,h_p$, so it can be simulated by $S$.
	\end{itemize}
	
	\item In the conversion protocol in Section \ref{conversion}:
	\begin{itemize}
		\item In Step 7, the orders of face-up cards are uniformly distributed among all $p!$ permutations of $\alpha_1,\alpha_2,...,\alpha_p$, so it can be simulated by $S$.
		\item In Step 10, the orders of face-up cards are uniformly distributed among all $p!$ permutations of $h_1,h_2,...,h_p$, so it can be simulated by $S$.
	\end{itemize}
	
	\item In the neighbor condition verification in Section \ref{neighbor}:
	\begin{itemize}
		\item In Step 3, the orders of face-up cards are uniformly distributed among all $m!$ permutations of $a_1,a_2,...,a_m$, so it can be simulated by $S$.
		\item In Step 4, the face-up card has an equal probability to be one of $b_2,b_3...,b_m$, so it can be simulated by $S$.
	\end{itemize}
	
	\item In the arrow condition verification in Section \ref{arrow}:
	\begin{itemize}
		\item In Step 3, the orders of face-up cards are uniformly distributed among all $(2m-1)!$ permutations of $a_1,a_2,...,a_{2m-1}$, so it can be simulated by $S$.
		\item In Step 4, the orders of face-up cards in Row 2 are uniformly distributed among all $\frac{(2m-2)!}{(m-2)!}$ permutations of $m$ cards selected from $b_2,b_3,...,b_{2m-1}$. The same goes for face-up cards in Row 3 and Row 4, with $m$ cards selected from $c_2,c_3,...,c_{2m-1}$ and $d_2,d_3,...,d_{2m-1}$, respectively. So, it can be simulated by $S$.
	\end{itemize}
\end{itemize}

Therefore, we can conclude that $V$ learns nothing about $P$'s colusion.
\end{proof}

\section{Future Work}
We developed a ZKP protocol for Makaro using a standard deck, which requires asymptotically less cards than the existing protocol of Bultel et al. \cite{makaro}. We also developed a general method to encode a number with a sequence of all different cards, which allows us to securely compute several numerical functions using a standard deck. This method can be used to verify solutions of some other logic puzzles including Suguru. Possible future work includes developing standard deck protocols to verify solutions of other logic puzzles (e.g. Kakuro, Numberlink), or to compute broader types of functions.

\end{document}